\documentclass[11pt]{amsart}

\pdfoutput = 1

\usepackage[english]{babel}
\usepackage[autostyle]{csquotes}
\usepackage{amsmath,amssymb,amsthm}
\usepackage{amsfonts}
\usepackage{amsxtra}
\usepackage{xcolor}

\usepackage{array,dcolumn}
\usepackage{graphicx}
\usepackage[hyperfootnotes=true,
        pdffitwindow=true,
        plainpages=false,
        pdfpagelabels=true,
        pdfpagemode=UseOutlines,
        pdfpagelayout=SinglePage,
                hyperindex,]{hyperref}

\usepackage{lmodern}
\usepackage[T1]{fontenc}
\usepackage[utf8]{inputenc}
\usepackage[activate={true,nocompatibility},spacing,kerning]{microtype}
\usepackage{bm}
\usepackage{bbm}
\usepackage[sc,osf]{mathpazo}
\microtypecontext{spacing=nonfrench}

  \calclayout

\newcommand{\mathsym}[1]{{}}
\newcommand{\unicode}[1]{{}}

\theoremstyle{plain}
\newtheorem{theorem}{Theorem}[section]

\newtheorem{corollary}[theorem]{Corollary}
\newtheorem{proposition}[theorem]{Proposition}

\theoremstyle{definition}

\theoremstyle{remark}
\newtheorem{remark}[theorem]{Remark}

\numberwithin{equation}{section}

\begin{document}


\title{Differential recurrences for the distribution of the trace of the $\beta$-Jacobi ensemble}
\author{Peter J. Forrester}
\address{School of Mathematics and Statistics, 
ARC Centre of Excellence for Mathematical \& Statistical Frontiers,
University of Melbourne, Victoria 3010, Australia}
\email{pjforr@unimelb.edu.au}

\author{Santosh Kumar}
\address{
Department of Physics, Shiv Nadar University, Uttar Pradesh 201314, India
}
 \email{skumar.physics@gmail.com}

\date{\today}


\begin{abstract}
Examples of the $\beta$-Jacobi ensemble specify the joint distribution of the transmission
eigenvalues in scattering problems. In this context, there has been interest in the distribution of the trace,
as the trace corresponds to the conductance. Earlier, in the case $\beta = 1$, the trace statistic
was isolated in studies of covariance matrices in multivariate statistics, where it is referred to as Pillai's $V$
statistic. In this context, Davis showed that for $\beta = 1$ the trace statistic, and its Fourier-Laplace transform,
can be characterised by $(N+1) \times (N+1)$ matrix differential equations. For the Fourier-Laplace transform,
this leads to a vector recurrence for the moments. However, for the distribution itself the characterisation provided
was incomplete, as the connection problem of determining the linear combination of Frobenius type solutions that
correspond to the statistic was not solved. We solve this connection problem for Jacobi parameter $b$ and
Dyson index $\beta$ non-negative integers. For the other Jacobi parameter $a$ also a non-negative integer, the power series portion of each Frobenius
solution terminates to a polynomial, and the matrix differential equation gives a recurrence for their computation.
\end{abstract}


\maketitle


\section{Introduction}\label{S1}
In random matrix theory a classical $\beta$-ensemble refers to an eigenvalue probability density function
proportional to
\begin{equation}\label{1}
\prod_{l=1}^N w(x_l) \prod_{1 \le j < k \le N} | x_k - x_j |^\beta,
\end{equation}
with $w(x)$ one of the classical weight functions from the theory of orthogonal polynomials,
\begin{equation}\label{eq:4.2}
		w(x)=
		\begin{cases}
		e^{-x^2}, \quad &\text{Gaussian}
		\\
		x^{a} e^{-x}\chi_{x>0}, \quad &\text{Laguerre}
		\\
		x^a (1-x)^b\chi_{0<x<1}, \quad &\text{Jacobi}. 		\end{cases}
	\end{equation}
Here the notation $\chi_A$ denotes the function taking the value $1$ for $A$ true and $0$ otherwise, and it is assumed
that the parameters $a,b > -1$ so that (\ref{1}) is normalisable.

In both the theory and applications of random matrix theory, the study of the distribution of a linear statistic
of the eigenvalues is prominent; see e.g.~\cite{PS11}. A linear statistic has the functional form
$\sum_{j=1}^N f(\lambda_j)$ --- thus $f$ is a function of a single eigenvalue, while the statistic itself is a symmetric
function of all the eigenvalues.  Our interest in this work is
in computational questions relating to the distribution of the linear statistic $\sum_{j=1}^N \lambda_j$ ---
which is the trace of the matrix --- for the $\beta$-Jacobi ensemble. In the case $\beta = 1$, when the $\beta$-Jacobi ensemble
relates to covariance matrices in multivariate statistics (see e.g.~\cite{Mu82}), the problem attracted interest as long
ago as 1955 \cite{Pi55}, as a test statistic for various correlation hypotheses, and it is referred to as 
Pillai's $V$ statistic or  Pillai's trace  (Pillai's name in this context is sometimes joined with
 Bartlett in recognition of the even earlier work \cite{Ba39}).
 
 More recently, the distribution of the traces for each of the cases $\beta = 1$, 2 and 4 of
 the $\beta$-Jacobi ensemble has attracted attention because of its relevance to the study of
 quantum transport problems in mesoscopic wires and cavities \cite{No08,KSS09,MS11}.
 The Jacobi parameters $a,b$ are determined by  the number of
 scattering channels for ingoing and outgoing waves, and the value of $\beta$; see
 e.g.~\cite[\S 1.2]{MS11}.
  Under the assumption that the scattering matrix is a uniformly distributed
 random (Haar measure) unitary matrix, respecting time reversal or spin-rotational
 symmetries where appropriate (such a symmetry determines $\beta$ --- usually referred to as the
 Dyson index), the $\{ x_j \}$ in (\ref{1}) correspond to the transmission eigenvalues and the trace (in dimensionless
 units) is the conductance; see the review \cite{Be97}. Physically, the case of uniformly distributed random unitary matrix corresponds to a situation when the two leads connected to the chaotic cavity are without any tunnel barriers; commonly referred to as ideal leads. This gives rise to the joint density of transmission eigenvalues coinciding with that of Jacobi ensemble with parameter $b=0$ (or, equivalently, $a=0$ for the reflection eigenvalues). The case of nonzero $b$ has been shown to arise for the transmission between two leads in a chaotic cavity comprising three leads~\cite{SM06}. Furthermore, the \emph{full} Jacobi weight (nonzero $a,b$) arises in the description of Andreev reflection eigenvalues in normal-metal-superconductor junction, and thermal transport in two-dimensional topological superconductors~\cite{Be15}.
We remark too that in the cases $\beta = 1$ and $2$, the cumulative distribution of the trace
 statistic occurs in the calculation of the volume of certain metric balls for real and complex Grassmann manifolds
 \cite{PWTC16,KPW16}.
 
 For the class of eigenvalue probability density functions (\ref{1}), the probability density
 function for the trace is specified by
\begin{equation}\label{3}
P(t) = {1 \over C_N} \int_I dx_1 \cdots \int_I dx_N \, \delta\Big ( t - \sum_{i=1}^N x_i \Big )
\prod_{l=1}^N w(x_l) \prod_{1 \le j < k \le N} | x_k - x_j |^\beta,
\end{equation} 
 where $C_N$ is the normalisation, $I$ is the interval of support of the eigenvalues and
 $\delta (y)$ is the Dirac delta function. The Fourier-Laplace transform of (\ref{3}), $\hat{P}(s)$ say, which can be viewed as
 the exponential generating function of the moments of the trace, is given by
\begin{equation}\label{3a}
\hat{P}(s) = {1 \over C_N} \int_I dx_1 \cdots \int_I dx_N \, 
\prod_{l=1}^N w(x_l) e^{-s x_l}  \prod_{1 \le j < k \le N} | x_k - x_j |^\beta,
\end{equation}  
and thus corresponds to a deformation of the weight with its multiplication by an exponential factor.
 In the Gaussian $(G)$ and Laguerre $(L)$ cases of (\ref{eq:4.2}) the
integral (\ref{3a}) is readily evaluated,
\begin{equation}
\hat{P}^{(G)}(s) = e^{Ns^2/4}, \qquad \hat{P}^{(L)}(s) = (1+s)^{-(a+1) N - \beta N (N - 1)/2}.
\end{equation}  
These are recognised as the characteristic functions of the normal distribution of mean zero, standard
deviation $\sqrt{N/2}$, and of the gamma distribution with parameters $k = (a+1)N +  \beta N (N - 1)/2$,
$\theta = 1$. However the computation of $P(t)$ in the Jacobi case is a far more complicated
task.

From the definitions
\begin{equation}\label{3b}
\hat{P}^{(J)}(s) = {1 \over S_N(a,b,\beta)} \int_0^1  dx_1 \cdots \int_0^1 dx_N \, 
\prod_{l=1}^N x_l^a (1 - x_l)^b e^{-s x_l}  \prod_{1 \le j < k \le N} | x_k - x_j |^\beta.
\end{equation}  
The normalisation $S_N(a,b,\beta)$ --- corresponding to the multiple integral in (\ref{3b}) with $s=0$ ---
is known as the Selberg integral and is specified by (\ref{S}) below. As to be discussed in
Section \ref{S2a}, (\ref{3b}) relates to a special function within the theory of multidimensional
hypergeometric functions based on Jack polynomials; see e.g.~\cite[Ch.~12]{Fo10}. From this theory,
the coefficient of $s^k$, $c_k^{(J)}$ say, can be given explicitly (Proposition \ref{P3}), but consists of a sum over all partitions $\kappa$
of no more than $N$ parts, and of length $k$. This is simple enough for small $k$ (e.g.~for $k=1$
there is only a single term corresponding to $\kappa = (1,0^{N-1})$ for $k=2$ there are two terms corresponding to the partitions of
two, $(2,0^{N-1})$ and $(1,1,0^{N-2})$ etc.), and the following formulas are obtained.

\begin{proposition}\label{P3a}
Let $m_k^{(J)}$ denote the $k$-th moment of the trace for the $\beta$-Jacobi ensemble, and set
\begin{equation}\label{u12}
u_1 = (\beta/2)(N-1) + a + 1, \qquad u_2 = \beta (N-1) + a + b + 2.
\end{equation}  
We have
\begin{equation}\label{u12+}
m_1^{(J)} = {N u_1 \over u_2}, \quad m_2^{(J)} = {N(N-1) \over 1 + \beta/2}{u_1 ( u_1 - \beta/2) \over u_2 ( u_2 - \beta/2)}  +
  {N (1 + \beta N /2) \over 1 + \beta/2}{u_1 ( u_1 + 1) \over u_2 ( u_2 + 1)},
 \end{equation}
 \begin{multline}\label{u12+3}
m_3^{(J)} = {N(N-1)(N-2) \over (1 + \beta/2)(1+\beta)} {u_1 ( u_1 - \beta/2) ( u_1 - \beta) \over u_2 ( u_2 - \beta/2) ( u_2 - \beta)}  +
3  {N(N-1)(1+\beta N/2)) \over (1 + \beta/4)(1+\beta)} \\ \times  {u_1 ( u_1 +1) ( u_1 - \beta/2) \over u_2 ( u_2 +1) ( u_2 - \beta/2)}  +
 {N(1 + \beta N/4)(1+\beta N/2) \over (1 + \beta/4)(1+\beta/2)} {u_1 ( u_1 +1) ( u_1 + 2) \over u_2 ( u_2 +1) ( u_2 + 2)}.
 \end{multline} 
\end{proposition}

However the number of terms in the sum specifying $m_k^{(J)}$ quickly proliferates, being asymptotically equal to  $e^{\pi \sqrt{2 k/3}}$,
for $k \le N$ large, as is well known from the theory of partitions. In Section \ref{S2b} a characterisation of $\hat{P}^{(J)}(s)$ in terms of a first order
linear matrix differential equation is given. One consequence
is a formula for $c_k^{(J)}$ in terms of a product of $(k+1)$ matrices
of size $(N+1) \times (N+1)$.

\begin{proposition}\label{P1}
Write $\mathbf c_k = [c_{p,k}]_{p=0}^N$. With $c_{N,0} = 1$,
specify $\{c_{s,0} \}_{s=0}^{N-1}$ by (\ref{MD4}) so that $\mathbf c_0$ is explicit.
With the $(N+1) \times (N+1)$ matrices $\mathbf X, \mathbf Y$ specified by (\ref{XY}), require that 
  \begin{equation}\label{MD5}  
\mathbf c_k  = \Big ( \prod_{s=1}^k ( (k-s+1)  \mathbf I_{N+1} - \mathbf Y)^{-1}  \Big )  \mathbf X  \mathbf c_0, \qquad (k=1,2,\dots).
 \end{equation}
 We have
  \begin{equation}\label{MD5a}  
  c_k^{(J)} = {(-1)^k m_k^{(J)}  \over k! } = (\mathbf c_k)_{N} =  (\mathbf c_k)_{0} \Big |_{b \mapsto b+1},
\end{equation}
where the notation $(\mathbf v)_{n}$ refers to the component in position $n$ of the vector $\mathbf v$,
with the numbering of positions starting at $0$. To avoid possible division by zero, only the final equality
should be used in the range $-1<b\le 0$.
\end{proposition}

The analytic structure of (\ref{3}) in the Jacobi case, $P^{(J)}(t)$ say, is more complex than that of its
Fourier-Laplace transform (\ref{3b}). It is supported on the interval $t \in [0,N]$ and has a distinct
functional form for each integer subinterval $[j,j+1]$ $(j=0,\dots,N-1)$ therein;
this latter property is well
known for the probability density function of fixed step length uniform random walks \cite{BS16}, and of that
for certain pattern avoiding permutations \cite{BEGM19} as examples from the recent literature.
 However, there are some
features in common. In both the subintervals $[0,1]$ and $[N-1,N]$ it is possible to use Jack polynomial
theory to give an expression for the coefficients in the power series expansion (about the origin for the
subinterval $[0,1]$, and about $t=N$ for the subinterval $[N-1,N]$) involving a sum over partitions of
fixed length; see Proposition \ref{P7} and (\ref{7e}). As noted in Proposition \ref{P8}, simplifications are
possible for $b=1$ and $b=-\beta/2$. Another common feature relates to a recursive computation of the
coefficients in series expansions relating to the points $t=p$ $(p=0,\dots,N-1)$ using a matrix differential
equation. This relies on identification of those series as the Frobenius type power series solution of 
the differential equation, and moreover on knowledge of the scalars in the required linear combination of
the latter (i.e.~solution of the connection problem), which is possible for $b$ and $\beta$ non-negative integers. 

\begin{proposition}\label{P1x}
Let $S_N(a,b,\beta)$ denote the Selberg integral as specified by (\ref{S}). Set 
  \begin{equation}\label{X1}
  \xi_p = (-1)^{p (b+1) + \beta p (p - 1)/2},
  \end{equation}
  define 
   \begin{equation}\label{10g} 
 \gamma_p = (N - 1) + b p + {\beta \over 2} p (p - 1) + a (N - p) +  {\beta \over 2} (N - p)  (N - p - 1),
 \end{equation}
 and specify $K_N(a,b,p,\beta)$ as in Proposition \ref{P10}. For $b, \beta$ non-negative integers we have
   \begin{equation}\label{10hy} 
   P^{(J)}(t) = {1 \over S_N(a,b,\beta)} \sum_{p=0}^{N-1} \xi_p \binom{N}{p}
   K_N(a,b,p,\beta) \chi_{p \le t \le N} (t - p)^{\gamma_p} F_N^{(p)} (t-p),
  \end{equation}
  where each $F_N^{(p)}(s)$ is analytic in the range $-1 < s < N-p$, normalised so that $F_N^{(p)}(0)=1$.
  
  If $a$ is further specialised to a non-negative integer, $F_N^{(p)}(s)$ reduces to a polynomial of degree
  (\ref{3.29a}). The latter is computable from a matrix differential equation (see \S 3.3).
  \end{proposition}  
  
  The matrix differential equation, of which $s^{\gamma_p}  F_N^{(p)}(s)$ is the final component of
  the Frobenius type vector solution, is specified by (\ref{GG+}) below. From this, a vector
  recurrence fully determining the coefficients in the power series $F_N^{(p)}(s)$ follows, which we state
  in Proposition \ref{P16}.

\section{The Fourier-Laplace transformed distribution}
 \subsection{Generalised hypergeometric function evaluation}\label{S2a}
The classical Gauss hypergeometric function admits the series form
 \begin{equation}\label{2.1}
 {}_2 F_1 (a,b;c;x) = \sum_{k=0}^\infty {(a)_k (b)_k \over k! (c)_k} x^k, \qquad (u)_k := {\Gamma(u+k) \over \Gamma(u)},
 \end{equation}
 and provides the integral evaluation
  \begin{equation}\label{2.2}
  {1 \over S_1(a,b,\cdot)} \int_0^1 x^a (1 - x)^b (1 - t x)^{-r} \, dx = {}_2 F_1 (r,a+1;a+b+2;t),
  \end{equation} 
  with the normalisation on the LHS referring to the $N=1$ case of the Selberg integral (\ref{S}) below,
  which is the Euler beta integral.
  Scaling $x \mapsto x/b$ in (\ref{2.1}) and taking the limit $b \to \infty$ shows
   \begin{equation}\label{2.3}
 \lim_{b \to \infty} {}_2 F_1 (a,b;c;x/b) = \sum_{k=0}^\infty {(a)_k \over k! (c)_k} x^k =: {}_1 F_1(a;c;x).
   \end{equation} 
 Making use of this result, we see by scaling $t \mapsto - t/r$ and taking $r \to \infty$ that
  \begin{equation}\label{2.4}
  {1 \over S_1(a,b,\cdot)} \int_0^1 x^a (1 - x)^b e^{-tx} \, dx = {}_1 F_1 (a+1,a+b+2;-t). 
  \end{equation} 
  
  In the theory of generalised hypergeometric functions based on Jack polynomials
  (see e.g.~\cite[Ch.~13]{Fo10}) there is a natural multidimensional generalisation of
  (\ref{2.1}) which allows for the evaluation of the case $p=0$ of the multiple integral
  (\ref{d1}). We only have need for this generalised hypergeometric function in the
  case that all the $n$ arguments are equal to the same value $x$ say; we write 
  $(x)^n$ as an abbreviation.
  
  Let $\alpha > 0$ be a scalar, let $\kappa = (\kappa_1,\dots, \kappa_n)$ be 
  a partition, and introduce the generalised Pochhammer symbol
    \begin{equation}\label{2.5}
    [u]_\kappa^{(\alpha)} = \prod_{j=1}^n {\Gamma(u - {1 \over \alpha} (j - 1) + \kappa_j) \over
   \Gamma(u - {1 \over \alpha} (j - 1) )}. 
   \end{equation}  
   Associated with a partition is its diagram, which records each nonzero part $\kappa_i$ as row
   $i$ of $\kappa_i$ boxes, drawn flush left starting from the first column. For a square $(i,j)$ in
   the diagram, the arm length $a(i,j)$ is defined as the number of boxes in the row to the
   right; the co-arm length $a'(i,j)$ is the number in the row to the left; the leg length $\ell(i,j)$ 
   is the number of boxes in the column below; the co-leg length $\ell'(i,j)$  is the number in the column
   and above the square.
   
   Let $C_\kappa^{(\alpha)}(x_1,\dots,x_n)$ denote the renormalised Jack polynomial of
  \cite[Def.~13.1.1]{Fo10}. The only property we require of this polynomial is its value
  when all arguments are unity,
  \begin{equation}\label{2.6} 
 C_\kappa^{(\alpha)}((1)^n) =  \alpha^{|\kappa|} |\kappa|!   {b_\kappa \over d_\kappa' h_\kappa},
  \end{equation}  
  where $b_\kappa, d_\kappa' , h_\kappa$ are specified in terms of the diagram of $\kappa$ according to
   \begin{equation}\label{2.6a} 
 b_\kappa = \prod_{(i,j) \in \kappa} \Big ( \alpha a'(i,j) + n -   \ell'(i,j) \Big ), \quad
d_\kappa'  =   \prod_{(i,j) \in \kappa} \Big ( \alpha (a(i,j) + 1) +   \ell(i,j) \Big )
  \end{equation}  
  and
  \begin{equation}\label{2.6b} 
 h_\kappa  =   \prod_{(i,j) \in \kappa} \Big ( \alpha a(i,j)  +   \ell(i,j) + 1\Big ).    
  \end{equation}    
 
 With this notation, the multidimensional generalisation of the Gauss hypergeometric function
 of interest is specified by
    \begin{equation}\label{2.8}
 {}_2^{} F^{(\alpha)}_1 (a,b;c;(x)^n)    = \sum_\kappa {   [a]_\kappa^{(\alpha)}   [b]_\kappa^{(\alpha)} \over   |\kappa|!   [c]_\kappa^{(\alpha)}} 
  C_\kappa^{(\alpha)}((1)^n) x^{|\kappa|}.
  \end{equation}
  Note that in the case $n=1$ (\ref{2.8}) reduces to (\ref{2.1}). The relevance of (\ref{2.8}) for present purposes is
  that it gives the evaluation of a multiple integral generalising (\ref{2.2}), of the type appearing in  (\ref{d1}).
  Thus \cite[Prop.~13.1.4]{Fo10}
  \begin{multline} \label{15.s1aa}
{1 \over S_N(\lambda_1,\lambda_2,1/\alpha)} 
\int^1_0 dx_1 \cdots \int^1_0 dx_N \, 
\prod_{l=1}^N x_l^{\lambda_1} (1-x_l)^{\lambda_2}
(1 - t x_l)^{-r}
\prod_{j<k}\left|x_j-x_k\right|^{2/\alpha}
 \\
 = {}_2^{} F_1^{(\alpha)}(r,
 {1 \over \alpha} (N-1) + \lambda_1 + 1;
{2 \over \alpha} (N-1) + \lambda_1 + \lambda_2 + 2;(t)^N).
\end{multline} 
We can now use the same limiting procedure that reduced 
 (\ref{2.1}) to  (\ref{2.3}) to obtain a generalised hypergeometric function of $\hat{P}^{(J)}(s)$. 
 
 \begin{corollary}\label{C2}
 In the notation of (\ref{2.8}), define
  \begin{equation}\label{2.9}
 {}_1^{}  F^{(\alpha)}_1 (a;c;(x)^n)    = \sum_\kappa {   [a]_\kappa^{(\alpha)} \over   |\kappa|!   [c]_\kappa^{(\alpha)}} 
  C_\kappa^{(\alpha)}((1)^n) x^{|\kappa|}.
  \end{equation} 
 With $\hat{P}^{(J)}(s)$ specified by (\ref{3b}), we have
 \begin{equation}\label{2.10} 
 \hat{P}^{(J)}(s) =  {}^{}_1  F^{(2/\beta)}_1 \Big ( (\beta/2) (n-1) + a + 1 ; \beta (n - 1) + a + b + 2 ;(-s)^n \Big ).
  \end{equation} 
  \end{corollary}
  
  \begin{proof}
  Analogous to (\ref{2.3}), it follows from (\ref{2.8}) and (\ref{2.5}) that
  $$
  \lim_{b \to \infty} {}^{{}^{}}_2 F_1^{(\alpha)} (a,b;c;(x/b)^n) =  {}_1^{}  F_1^{(\alpha)}(a;c;(x)^n ).
  $$
  Consequently, we see that by replacing $ t $ by $-s/r$ in (\ref{15.s1aa}) and taking the limit
  $ r \to \infty$ the evaluation (\ref{2.10}) results, upon an appropriate change of notation.
  \end{proof}

  The result of Corollary \ref{C2} makes immediate a formula for the $k$-th moment of  $P^{(J)}(t)$.
  
  \begin{proposition}\label{P3}
  The $k$-th moment of the trace for the $\beta$-Jacobi ensemble is given by
  \begin{equation}\label{2.11}   
 m_k^{(J)} = \Big \langle \Big ( \sum_{j=1}^N x_j \Big )^k \Big \rangle^{(J)} = \sum_{\kappa: \, |\kappa| = k}
  {[u_1]_\kappa^{(2/\beta)} \over [u_2]_\kappa^{(2/\beta)} }   C_\kappa^{(2/\beta)}((1)^N) ,
  \end{equation}
 where $ C_\kappa^{(2/\beta)}((1)^N) $ is specified by  (\ref{2.6}) and
  \begin{equation}\label{2.12}    
  u_1 = (\beta/2) (N - 1) + a + 1, \qquad   u_2 = \beta (N - 1) + a + b +2.
  \end{equation}
  \end{proposition}
  
  Evaluating the right hand side of (\ref{2.11}) for $k=1, 2$ and $3$ gives the result of Proposition \ref{P3a}.	
  
   \begin{remark}
In the case $\beta = 2$ and $b=0$ the result of Proposition \ref{P3} was first given by Novaes \cite[Eq.~(24)]{No08},
while the result for $\beta = 1$ and $b=0$, albeit written in a more complicated form, was given
by  Khoruzhenko et al.~\cite[Eqns.~(15), (16) \& 22]{KSS09}. For general $\beta,a,b$, the explicit formula for
$m_1^{(J)}$ as listed in Proposition \ref{P3a} can be found in the earlier work \cite[Eq.~B.7a]{MRW15}.
\end{remark}

Define
$$
\mu_k^{(J)} = \mu_k^{(J)}(N,\beta,a,b) = \Big \langle \sum_{j=1}^N x_j^k \Big \rangle^{(J)}.
$$
As with $m_k$, this quantity is a rational function of the parameters --- see \cite{FLD16,MRW15}. It
satisfies the functional equation \cite{DP12,FLD16,FRW17}
 \begin{equation}\label{2.12a}  
 \mu_k^{(J)}(N,\beta,a,b) = - (2/\beta)  \mu_k^{(J)}(-\beta N/2,4/\beta,-2a/\beta,-2b/\beta).
 \end{equation}
 In the case $k=1$, $m_k^{(J)}$ and $\mu_k^{(J)}$ coincide, and so $m_k^{(J)}$ also
 satisfies the functional equation (\ref{2.12a}) in this case. In fact, up to the normalising
 factor, $m_k^{(J)}$ satisfies the functional
 equation (\ref{2.12a}) for general $k$.
 
  \begin{proposition}\label{P3aa}
  Define $m_k^{(J)}$ by the average (\ref{2.11}). A functional equation of the form (\ref{2.12a}) holds,
   \begin{equation}\label{2.13a}  
 m_k^{(J)}(N,\beta,a,b) =  (-2/\beta)^k  m_k^{(J)}(-\beta N/2,4/\beta,-2a/\beta,-2b/\beta).
 \end{equation}
 \end{proposition}
 
 \begin{proof}
 As specified below (\ref{2.5}), associated with each partition $\kappa$ is a diagram.
 Reversing the role of rows and columns in the diagram gives the conjugate partition,
 denoted by $\kappa'$. The length of the partition does not change, so in the summand
 of (\ref{2.11}) it is permissible to replace $\kappa$ by $\kappa'$. We do this for the RHS
 of (\ref{2.13a}). Now in general
 \cite[Exercises 12.4 q.2]{Fo10}
 $$
 [a]_{\kappa'}^{(1/\alpha)} = ( - \alpha)^{-|\kappa|} [-a/\alpha]_{\kappa}^{(\alpha)} .
 $$
 Furthermore, from the definition (\ref{2.6}) we see that
 $$
  C_{\kappa'}^{(1/\alpha)}((1)^N) \Big |_{N \mapsto - N/\alpha}
  = (- \alpha)^{-|\kappa|} C_\kappa^{(\alpha)}((1)^N) .
  $$
  Using these formulas with $\alpha = 2/\beta$ reduces the RHS of (\ref{2.13a}) to the summation
  of (\ref{2.11}).
 
 \end{proof}

 \subsection{The differential-difference system}\label{S2b}
 Let $C_p^N$ denote the binomial coefficient, let $e_p(y_1,\dots,y_N)$ denote the $p$-th elementary
 symmetric polynomial in $\{y_i \}_{i=1}^N$ and define
 \begin{multline}\label{5.1}
J_{p,N}^{(\alpha)}(x) = {1 \over C_p^N}\int_0^x dt_1 \cdots  \int_0^x dt_N \,
\prod_{l=1}^N t_l^{\lambda_1} (1 - t_l)^{\lambda_2} (x - t_l)^\alpha \\
\times \prod_{1 \le j < k \le N} | t_k - t_j|^\beta 
e_p(x - t_1,\dots, x - t_N).
\end{multline}
We know from \cite{Fo93}, \cite[\S 4.6.4]{Fo10} 
that  this family of multiple integrals satisfies the
differential-difference system
\begin{multline}\label{5.1a}
(N - p) E_p J_{p+1}(x) 
= (A_p x + B_p) J_p(x) - x(x-1) {d \over dx} J_p(x) + D_p x ( x - 1) J_{p-1}(x),
\end{multline}
valid for $p=0,\dots,N$ and where we have abbreviated $J_{p,N}^{(\alpha)}(x) =: J_p(x)$, and
\begin{align*}
A_p & = (N-p) \Big ( \lambda_1 + \lambda_2 + \beta (N - p - 1) + 2(\alpha + 1) \Big ) \\
B_p & = (p-N)  \Big ( \lambda_1 + \alpha + 1 + (\beta/2) (N - p - 1) \Big ) \\
D_p & = p \Big ( (\beta/2)(N-p) + \alpha + 1 \Big ) \\
E_p & = \lambda_1 + \lambda_2 + 1 + (\beta/2) (2N - p - 2) + (\alpha + 1).
\end{align*}
A recent application of the differential-difference system to the computation of some structured formulas for the distribution of the smallest
and largest eigenvalues of the $\beta$-Jacobi ensemble, applicable in certain parameter ranges, has
been given in \cite{FK20}. In the case of the $\beta$-Laguerre ensemble
analogue of (\ref{5.1}), the works \cite{Ku19,FK19,FT19} exhibit analogous applications.
For present purposes, a particular change of variables and
limit of (\ref{5.1}) and (\ref{5.1a}) is
required.

\begin{corollary}\label{C6}
For $p=0,1,\dots,N$ define
\begin{multline}\label{6.1}
\hat{H}_{p}(x) = {1 \over C_p^N}\int_0^1 dt_1 \cdots  \int_0^1 dt_N \,
\prod_{l=1}^N t_l^{a} (1 - t_l)^{b-1} e^{-x \sum_{l=1}^N t_l}  \\
\times \prod_{1 \le j < k \le N} | t_k - t_j|^\beta 
e_p(1 - t_1,\dots, 1 - t_N).
\end{multline}
With  the notation
\begin{align*}
\tilde{B}_p & = p\Big ( a + b + 1 + (\beta/2)(2N - p - 1) \Big ) \\
\tilde{D}_p & = p \Big ( (\beta/2) (N - p) + b \Big ),
\end{align*}
 we have that $\{  \hat{H}_p(x) \}_{p=0}^{N}$ satisfies the
differential-difference system
\begin{equation}\label{6.1a}
(N - p) x \hat{H}_{p+1}(x) 
= \Big ( (N - p) x + \tilde{B}_p \Big ) \hat{H}_p(x)  +  x  {d \over dx} \hat{H}_p(x) -  \tilde{D}_p  \hat{H}_{p-1}(x), \quad (p=0,\dots, N).
\end{equation}
\end{corollary}

\begin{proof}
Changing variables $t_l \mapsto x t_l$ in (\ref{5.1}) shows that
$$
J_{p,N}^{(\alpha)}(x) = x^{(\lambda_1 + \alpha + 1)N + p + \beta N (N - 1)/2 } \hat{J}_{p,N}^{(\alpha)}(x),
$$
where
 \begin{multline}\label{d1}
\hat{J}_{p,N}^{(\alpha)}(x) = {1 \over C_p^N}\int_0^1 dt_1 \cdots  \int_0^1 dt_N \,
\prod_{l=1}^N t_l^{\lambda_1} (1 - x t_l)^{\lambda_2} (1 - t_l)^\alpha  \\
\times \prod_{1 \le j < k \le N} | t_k - t_j|^\beta 
e_p(1 - t_1,\dots, 1 - t_N).
\end{multline}
Substituting in (\ref{5.1a}) shows that
\begin{multline}\label{7.1a}
(N - p) x E_p \hat{J}_{p+1}(x) 
= (\hat{A}_p x + \hat{B}_p) \hat{J}_p(x) - x(x-1) {d \over dx} \hat{J}_p(x) + D_p  ( x - 1) \hat{J}_{p-1}(x),
\end{multline}
where
\begin{align*}
\hat{A}_p & = A_p - \Big ( (\lambda_1 + \alpha + 1)N  + p + \beta N (N - 1)/2 \Big )\\
\hat{B}_p & = B_p + \Big ( (\lambda_1 + \alpha + 1)N  + p + \beta N (N - 1)/2 \Big ).
\end{align*}
Now replace $x$ by $x/\lambda_2$ and take the limit $\lambda_2 \to \infty$. We see
that with $\lambda_1 = a, \alpha =  b - 1$, $\hat{J}_{p,N}^{(\alpha)}(x)  \to  \hat{H}_{p}(x) $
and the recurrence (\ref{7.1a}) reduces to (\ref{6.1a}).
\end{proof}

\begin{remark}
This result in the case $\beta = 1$ can be found in the work of Davis \cite{Da70}, in the context of a study of
the distribution of Pillai's trace. Davis had earlier \cite{Da68} introduced matrix differential equations in an
analysis of a closely related statistic in multivariate statistics --- Hotelling's generalised $T_0^2$ .
\end{remark}

Introduce the  $(N+1) \times (N+1)$ matrices 
\begin{align}\label{XY}
\mathbf X & =  - {\rm diag} \, [N, N - 1,\dots,0] + {\rm diag}^+ \,[N,N-1,\dots,1]  \nonumber  \\
\mathbf Y & = - {\rm diag} \, [\tilde{B}_0, \tilde{B}_1,\dots, \tilde{B}_N] + {\rm diag}^- \, [\tilde{D}_1, \tilde{D}_2,\dots, \tilde{D}_N],
\end{align}
where ${\rm diag}^+$ refers to the first diagonal above the main diagonal, and
${\rm diag}^-$ the first diagonal below; all other entries are zero. With $\hat{\mathbf H}(x) = [\hat{H}_p(x)]_{p=0}^N$,
following the works of Davis and also \cite{FR12},
we see that the differential-difference equation (\ref{6.1a}) is equivalent to the matrix differential equation
\begin{equation}\label{MD}
x {d \over dx} \hat{\mathbf H}(x) = (x \mathbf X + \mathbf Y)  \hat{\mathbf H}(x).
\end{equation}
Denote
\begin{align}\label{S}
S_N(a,b,\beta) & := \int_0^1 dt_1 \cdots  \int_0^1 dt_N \,
\prod_{l=1}^N t_l^{a} (1 -  t_l)^{b}  \prod_{1 \le j < k \le N} | t_k - t_j|^\beta  \nonumber \\
& = \prod_{j=0}^{N-1}\frac{\Gamma(a+1+j\beta/2)\Gamma(b+1+j\beta/2)\Gamma(1+(j+1)\beta/2)}{\Gamma(a+b+2+(N+j-1)\beta/2)\Gamma(1+\beta/2)},
 \end{align}
where the product of gamma function evaluation is due to Selberg \cite{Se44}. 
According to (\ref{6.1}), we want the solution of (\ref{MD}) having the power series form
\begin{equation}\label{MD1}
 \hat{\mathbf H}(x) = \Big [ \sum_{l=0}^\infty c_{p,l} x^l \Big ]_{p=0}^N,
 \end{equation}
 and with
 \begin{equation}\label{MD2}
  c_{0,0} = S_N(a,b-1,\beta), \qquad  c_{N,0} = S_N(a,b,\beta).
\end{equation}

In fact specifying just the latter value in (\ref{MD2}), all coefficients $\{ c_{p,l} \}$ in (\ref{MD1})
are determined by the matrix differential equation (\ref{MD}), and this allows for Proposition
\ref{P1} to be established.

\medskip
\noindent
\emph{Proof of Proposition \ref{P1}.} 
With
$\mathbf c_l := [ c_{p,l} ]_{p=0}^N$, substituting (\ref{MD1}) in (\ref{MD}) and equating like powers of
$x$ gives the first order vector recurrence
 \begin{equation}\label{MD3}
 (l {\mathbf I}_{N+1} - \mathbf Y) \mathbf c_l  = \mathbf X  \mathbf c_{l-1}, \qquad \mathbf c_{-1} := \mathbf 0,
 \end{equation}
 where $\mathbf{I}_{N+1}$ denotes the $(N+1) \times (N+1)$ identity.
 Moreover, in the case $l=0$, we see from the definition of $\mathbf Y$ that
  \begin{equation}\label{MD4}
   c_{N - l,0} =  c_{N,0}  \prod_{s=1}^l {\tilde{B}_{N+1-s} \over \tilde{D}_{N+1-s}}, \qquad (l=1,\dots,N).
   \end{equation}
   Evaluation of (\ref{MD4}) with $l = N$ gives consistency with (\ref{MD2}).
With $\mathbf c_0$ determined by (\ref{MD4}), iterating (\ref{MD3}) gives
(\ref{MD5}) above. For the latter  to make sense we require that $(s \mathbf I_{N+1} - \mathbf Y)$ be invertible for each
 $s=1,2,\dots$; this is immediate from the definition of $\mathbf Y$.
 After normalisation by (\ref{S}), and so  setting $c_{N,0} = 1$ in (\ref{MD4}) rather than its value in (\ref{MD2}),
 we see upon comparing (\ref{3b}) and (\ref{6.1})
 that the final component in $\mathbf c_l $
 corresponds to $c_l^{(J)}$, as does the first component but with
 $b$ replaced by $b-1$, establishing the result. Note that due to the first
 of the Selberg integrals in (\ref{MD2}) requiring $b>0$ to be well defined for
 all $\beta > 0$, for the range $-1<b\le0$ only the final equality in (\ref{MD5a}) should be used.
  \hfill $\square$
  
Based on the recursion in (\ref{MD3}), we provide a Mathematica~\cite{Mathematica} code as an ancillary file to evaluate the coefficients $c_{p,l}$ which is used to obtain the power-series for $\hat{\mathbf H}(x)$ to the desired order and also the moments $m_k^{(J)}$ using the relation (\ref{MD5a}).
 
 \begin{remark}
1.~ The average (\ref{3b}) specifying the Fourier-Laplace transform of the trace statistic for the
 $\beta$-Jacobi ensemble has, for specialisation of the parameters, other interpretations in random
 matrix theory. As one example, let $E_{N,\beta}^{(L)}(0,(s,\infty);a)$ denote the probability that 
 the interval $(s,\infty)$ is free of eigenvalues in the $\beta$-Laguerre ensemble, specified by the
 eigenvalue probability density function (\ref{1}) with Laguerre weight (\ref{eq:4.2}). A simple change
 of variables in the definition shows
  \begin{equation}\label{LD4}
 E_{N,\beta}^{(L)}(0,(s,\infty);a) = {S_N(a,0,\beta) \over L_{N}(a,\beta)} s^{(a+1)N + \beta N (N - 1)/2} \hat{P}^{(J)}(s) \Big |_{b=0},
 \end{equation}
 where 
    \begin{align}\label{LD5}
L_{N}(a,\beta) & := \int_0^\infty dx_1 \cdots      \int_0^\infty dx_N \, \prod_{l=1}^N x_l^a e^{-x_l} \prod_{1 \le j < k \le N}
|x_k - x_j|^\beta \nonumber \\
& = \prod_{j=0}^{N-1} {\Gamma(1 + (j+1)\beta/2) \Gamma(a +1 + j \beta/2) \over \Gamma(1+\beta/2)},
\end{align}
with the gamma function evaluation following as a limiting case of the Selberg integral (\ref{S}); see
e.g.~\cite[Prop.~4.7.3]{Fo10}.

Another example relates to the gap probability $E_{N,\beta}^{(L)}(0,(0,s);a)$ --- that is the probability there
are no eigenvalues in the interval $(0,s)$ --- in the particular cases $\beta = 1$, $a \in \mathbb Z_{\ge 0}$ and
and
$\beta = 4$, $a \in \mathbb 2 Z_{\ge 0}$.
Let $E_1^{\rm hard}(0,(0,s);a)$ denote the corresponding hard edge scaled limit; see e.g.~\cite[\S 9.8]{Fo10}
for the precise definition of this limit. The point of interest in the present context are the formulas
\cite{FW02c,FW04}
  \begin{equation}\label{LD6a}
 E_1^{\rm hard}(0,(0,s^2);m) = e^{-s^2/8 + m s} \hat{P}^{(J)}(2s) \Big |_{N=m, \beta =2 \atop a=b=1/2}
  \end{equation}
  and
  \begin{equation}\label{LD6b}
 E_4^{\rm hard}(0,(0,s^2);m) = e^{-s^2/8 + m s} \Big ( \hat{P}^{(J)}(2s) \Big |_{N=m, \beta =2 \atop a=b=1/2} +
 e^s  \hat{P}^{(J)}(2s) \Big |_{N=m+1, \beta =2 \atop a=b=-1/2} \Big ).
  \end{equation}  
  
 Since both $ E_1^{\rm hard}$ and $ E_4^{\rm hard}$ evaluate particular probabilities in symmetrised versions of
 Hammersley's model of directed paths of maximum length for random points in the square with a Poisson distribution
 (see e.g.~\cite[\S 10.6]{Fo10}),
 the formulas (\ref{LD6a}) and (\ref{LD6b}) give an interpretation of particular instances of the trace statistic for the
 Jacobi unitary ensemble (the case $\beta = 2$ of the  $\beta$-Jacobi ensemble) outside of random matrix theory
 too.
 
 2.~The case $\beta = 2$ of (\ref{3b}) permits an evaluation in terms of a transcendent from the Hamiltonian theory
 of Painlev\'e V, which satisfies the so-called $\sigma$ PV equation \cite{FW02b}. This provides a scheme to
 compute the corresponding power series in terms of nonlinear recurrences.
  \end{remark}

 \section{The distribution of the trace}
 \subsection{The range $0 \le t < 1$}
 Changing variables $x_i \mapsto t x_i$ ($i=1,\dots,N$) in
 (\ref{3}) with the Jacobi weight shows that for $0 \le t < 1$
 \begin{multline}\label{7b}
{P}^{(J)}(t) = {t^{(a+1)N + \beta N (N - 1)/2-1}  \over S_N(a,b,\beta)} \\
\times  \int_0^\infty  dx_1 \cdots \int_0^\infty dx_N \,  \delta \Big ( 1 - \sum_{i=1}^N x_i \Big )
\prod_{l=1}^N x_l^a (1 - t x_l)^b  \prod_{1 \le j < k \le N} | x_k - x_j |^\beta.
\end{multline}  
The integrand in (\ref{3b}) with $t=0$ is recognised as proportional to the
probability density function specifying the so-called fixed trace Laguerre
ensemble (see e.g.~\cite[Eq.~(3.56)]{Fo10}). The latter has the known
normalisation \cite[Eq.~(B2)]{KSS09}
 \begin{align}\label{7c}
 F_N(a,\beta)  & :=  \int_0^\infty  dx_1 \cdots \int_0^\infty dx_N \,  \delta \Big ( 1 - \sum_{i=1}^N x_i \Big )
\prod_{l=1}^N x_l^a   \prod_{1 \le j < k \le N} | x_k - x_j |^\beta \nonumber  \\ & =
{W_{a,\beta,N} \over \Gamma((a+1)N + \beta N (N - 1)/2)}, 
 \end{align}
 where
  \begin{align}\label{7c+}
  W_{a,\beta,N} & := \int_0^\infty  dx_1 \cdots \int_0^\infty dx_N \,  \prod_{l=1}^N x_l^a  e^{-x_l }
 \prod_{1 \le j < k \le N} | x_k - x_j |^\beta  \nonumber \\ &=
 \prod_{j=0}^{N-1} {\Gamma(a +1 + \beta j /2)  \Gamma(1 + \beta (j+1) /2)  \over \Gamma (1 + \beta / 2)}.
  \end{align}
It thus follows from (\ref{7b}) that for $0 \le t \le 1$
 \begin{equation}\label{7d}
{P}^{(J)}(t) = {  F_N(a,\beta) \over  S_N(a,b,\beta)} t^{(a+1)N + \beta N (N - 1)/2-1} \sum_{p=0}^\infty \alpha_p^{(J)} t^p,
 \end{equation}
 with $\alpha_0^{(J)} = 1$. The sum has upper terminal $bN$ in the case $b \in \mathbb Z_{\ge 0}$, when the
 multidimensional integral in (\ref{7b}) is a polynomial in $t$ of degree $bN$, and in particular
  \begin{equation}\label{7d+}
{P}^{(J)}(t) \Big |_{b=0} = {  F_N(a,\beta) \over  S_N(a,b,\beta)} t^{(a+1)N + \beta N (N - 1)/2-1} .
 \end{equation}
 
 Denote the fixed trace Laguerre ensemble as specified by the probability density corresponding
 to the integrand of (\ref{7c}), by fL$\beta$E${}_a$. Similarly, denote the $\beta$-Laguerre ensemble,
 specified by the probability density proportional to (\ref{1}) with Laguerre weight as in (\ref{eq:4.2}), by
 L$\beta$E${}_a$. We know from \cite[Ex.~3.3 q.4]{Fo10} that for a homogeneous polynomial
 of the eigenvalues of degree $|\kappa|$,
  \begin{equation}\label{3.1u}
  \langle p_\kappa\rangle_{{\rm fL}\beta{\rm E}_a} = {1 \over (\beta N (N - 1)/2 + N (a+1))_{| \kappa |}}
    \langle p_\kappa \rangle_{{\rm L}\beta{\rm E}_a}.
  \end{equation}
  We can make use of (\ref{3.1u}), together with results from the theory of Jack polynomials \cite[Ch.~12 \& 13]{Fo10}
  to give a formula for the $\alpha_p^{(J)}$ in (\ref{7d}).
  
  \begin{proposition}\label{P7}
  In the expansion (\ref{7d}) of ${P}^{(J)}(t)$ for $0 \le t \le 1$ we have
   \begin{equation}\label{7e+}
  \alpha_p^{(J)} = {1 \over
 p!  (\beta N (N - 1)/2 + N (a+1))_p } \sum_{\kappa: \, |\kappa|=p}
 [-b]_\kappa^{(2/\beta)} [a+\beta(N-1)/2+1]_\kappa^{(2/\beta)} C_\kappa^{(2/\beta)}((1)^N).
  \end{equation}
  \end{proposition}
  
  \begin{proof}
  The generalised binomial theorem from the theory of Jack polynomials tells us that \cite[Eqns.~(12.133) \& (13.1)]{Fo10}
  \begin{equation}\label{3.1x}   
  \prod_{l=1}^N(1 - t x_l)^b = \sum_{\kappa} {[-b]_\kappa^{(\alpha)} \over | \kappa|!} t^{|\kappa|}
  C_\kappa^{(\alpha)}(x_1,\dots,x_N),
  \end{equation}  
  where $\alpha > 0$ is arbitrary. This same line of theory also tells us that \cite[Eq.~(12.153)]{Fo10}
  \begin{equation}\label{3.1y}  
   \langle C_\kappa^{(2/\beta)}(x_1,\dots,x_N) \rangle_{{\rm L}\beta{\rm E}_a} =  C_\kappa^{(2/\beta)}((1)^N) [a + \beta (N - 1)/2 + 1]_\kappa^{(2/\beta)}.
  \end{equation}
  Using (\ref{3.1x}) with $\alpha = 2/\beta$, (\ref{3.1y}) and (\ref{3.1u}) in the integral of (\ref{7b}) shows
  \begin{multline*}
  {1 \over F_N(a,\beta)} \int_0^\infty dx_1 \cdots  \int_0^\infty dx_N \, \delta \Big ( 1 
  - \sum_{i=1}^N x_i \Big )
\prod_{l=1}^N x_l^a (1 - t x_l)^b  \prod_{1 \le j < k \le N} | x_k - x_j |^\beta \\
=  \sum_{\kappa}
{ t^{|\kappa|}  [-b]_\kappa^{(2/\beta)} [a+\beta(N-1)/2+1]_\kappa^{(2/\beta)} \over |\kappa|!  (\beta N (N - 1)/2 + N (a+1))_{| \kappa |}} 
 C_\kappa^{(2/\beta)}((1)^N),
 \end{multline*}
    and (\ref{7e+}) follows.
    \end{proof}
    
    Suppose we know $\{ \beta_p \}$ in the power series expansion
    $$
    \Big \langle  \prod_{l=1}^N  (1 - t x_l )^b  \Big \rangle_{{\rm L}\beta{\rm E}_a} = \sum_{p=0}^\infty \beta_p t^p.
    $$
    The above proof tells us that then
   \begin{equation}\label{4.1x}      
    \alpha_p^{(J)} = {1 \over
  (\beta N (N - 1)/2 + N (a+1))_p }  \beta_p.
  \end{equation}
  This viewpoint allows for a simplification of Proposition \ref{P7} in some special cases.
  
  \begin{proposition}\label{P8}  
  As an alternative to the formula (\ref{7e+}, we have in the case $b=1$
  \begin{equation}\label{4.1y}   
   \alpha_p^{(J)} = {(-N)_p (-(N-1) - (2/\beta) (a+1))_p \over p!   
     (\beta N (N - 1)/2 + N (a+1))_p} \Big ( - {\beta \over 2} \Big )^p
 \end{equation}
 (note that this vanishes for $p > N$), and in the case $b = - \beta / 2$ (this requires $\beta < 2$ 
 for the Jacobi $\beta$-ensemble to be normalisable)
   \begin{equation}\label{4.1z}   
   \alpha_p^{(J)} = {(\beta N/2)_p ( (\beta/2)(N-1) + (a+1))_p \over p!   
     (\beta N (N - 1)/2 + N (a+1))_p}.
 \end{equation}
 \end{proposition}     
 
  \begin{proof}
  Taking $m=1$, writing $t \mapsto \lambda_2 t$ and taking the limit $\lambda_2 \to \infty$
  in \cite[Eq.~(13.7)]{Fo10} shows
    \begin{equation}\label{4.1a}  
   \Big \langle  \prod_{l=1}^N  (1 - t x_l )^b \Big \rangle_{{\rm L}\beta{\rm E}_a} = \sum_{p=0}^N
   {(-N)_p (-(N-1) - (2/\beta)(a+1))_p \over p! } \Big ( - {\beta t \over 2} \Big )^p.
  \end{equation}  
   Also, taking  $m=1$, writing $t \mapsto \lambda_2 t$ and taking the limit $\lambda_2 \to \infty$
  in \cite[Eq.~(13.10)]{Fo10} shows
 \begin{equation}\label{4.1aa}  
   \Big \langle \prod_{l=1}^N (1 - t x_l )^{-\beta/2} \Big \rangle_{{\rm L}\beta{\rm E}_a} \doteq \sum_{p=0}^\infty
   {( \beta N/2)_p ((\beta/2)(N-1) + a + 1)_p \over p! } t^p.
  \end{equation}    
  Here the use of $\doteq$ indicates that both sides are to be considered as formal power series in $t$.
  Now applying (\ref{4.1x}) gives the stated results.
  \end{proof}
  
  Substituting (\ref{4.1y}) and (\ref{4.1z}) in (\ref{7d}) shows that in these cases $P^{(J)}(t)$ can
  be expressed in terms of the Gauss hypergeometric function.
   \begin{corollary} 
For $0 < t < 1$ we have
   \begin{multline}\label{9d}
{P}^{(J)}(t) \Big |_{b=1} = {  F_N(a,\beta) \over  S_N(a,1,\beta)} t^{(a+1)N + \beta N (N - 1)/2-1}  \\
\times {}_2 F_1\Big (-N, -(N-1) - (2/\beta) (a+1);\beta N (N - 1)/2 + N (a+1);-\beta t/2 \Big )
 \end{multline} 
 and
  \begin{multline}\label{9e}
{P}^{(J)}(t) \Big |_{b=-\beta/2} = {  F_N(a,\beta) \over  S_N(a,-\beta/2,\beta)} t^{(a+1)N + \beta N (N - 1)/2-1} \\
\times {}_2 F_1\Big ( \beta N/2, (\beta/2)(N-1) + (a+1);\beta N (N - 1)/2 + N (a+1);t \Big ).
 \end{multline} 
\end{corollary}
  
 \subsection{General range $0 \le t \le N$ and proof of Proposition \ref{P1x}} 
 The definition of ${P}^{(J)}(t)$ shows that it has support on $0 \le t \le N$, and displays the symmetry
 \begin{equation}\label{7e} 
{P}^{(J)}(N -t) = {P}^{(J)}(t) \Big |_{a \leftrightarrow b}.
\end{equation}
Hence from (\ref{7d}), or (\ref{9d}) , (\ref{9e}) for $a=-\beta/2$, $a=1$,
we can read off the form of the series expansion about $t = N$, valid for
$N-1 < t \le N$.

For general $t > 1$ it is convenient to rewrite the range of integration $[0,1]$ of
each integration in the definition (\ref{3}) of $P^{(J)}(t)$ according to the manipulation
 \begin{align}\label{M1e}
\int_0^1 dx_1 \cdots \int_0^1 dx_N & = \Big ( \int_0^t  - \int_1^t \Big )  dx_1 \cdots \Big ( \int_0^t  - \int_1^t \Big ) dx_N  \nonumber \\
& = 
 \sum_{p=0}^N  \binom{N}{p} \int_1^t dx_1 \cdots \int_0^t dx_p \int_0^1 dx_{p+1} \cdots \int_0^1 dx_N 
\end{align}
 with the second equality valid whenever the integrand is symmetric.
 Requiring that $b$ be a non-negative integer, this gives
  \begin{multline}\label{10e}
{P}^{(J)}(t)  = {1 \over S_N(a,b,\beta)}\bigg ( \int_0^t dx_1 \cdots \int_0^t dx_N  -  \sum_{p=1}^N  \binom{N}{p} 
\int_1^t dx_1 \cdots \int_1^t dx_p \\ \times \int_0^t dx_{p+1} \cdots \int_0^t dx_N 
\bigg ) \delta \Big ( t - \sum_{l=1}^N x_l \Big )
\prod_{l=1}^N x_l^a (1 - x_l)^b \prod_{1 \le j < k \le N} | x_k - x_j |^\beta.
\end{multline}
If we were to replace each $(1 - x_l)^b$ by $|1 - x_l|^b$ this would be true without requiring that
$b \in \mathbb Z_{\ge 0}$. However this would change the analytic structure as a function of $t$,
as we will see subsequently.

The delta function constraint tells us that the support of the $p$-th term is
$t \ge p$ and that the support of each integration variable 
$x_i$ $(i=1,\dots,p)$ is $(1,t-p+1)$, while that of $x_j$ $(j=p+1,\dots,N)$
is $(0,t-p)$. Note that since the support of ${P}^{(J)}(t)$ is $0 \le t \le N$,
the term $p=N$ does not contribute. Considering these features,  a simple change of variables in (\ref{10e})
gives a sum of the form (\ref{10hy}) with
 \begin{equation}\label{7ed} 
\xi_p = (-1)^{(b+1)p} 
 \end{equation}
 and
 \begin{multline}\label{10fx}
F_N^{(p)}(s) =  {1 \over  K_N(a,b,p,\beta)} \int_0^1 dx_1 \cdots \int_0^1 dx_p \,\\
\times \prod_{l=1}^p (1 + s x_l)^a  x_l^b \prod_{1 \le j < k \le p} | x_k - x_j |^\beta 
 \int_0^1 dx_{p+1} \cdots \int_0^1 dx_N  \, \prod_{l=p+1}^N  x_l^a  ( 1 - s x_l )^b \\ \times \prod_{p+1 \le j < k \le N} | x_k - x_j |^\beta 
 \prod_{l=1}^p \prod_{l'=p+1}^N | 1 + s (x_l - x_{l'}) |^\beta   \delta \Big ( 1 - \sum_{l=1}^N x_l \Big ),
\end{multline}
where 
 \begin{multline}\label{10h} 
 K_N(a,b,p,\beta):=  \int_0^1 dx_1 \cdots \int_0^1 dx_p  \, \prod_{l=1}^p   x_l^b \prod_{1 \le j < k \le p} | x_k - x_j |^\beta \\
\times \int_0^1 dx_{p+1} \cdots \int_0^1 dx_N  \, \prod_{l=p+1}^N  x_l^a  \prod_{p+1 \le j < k \le N} | x_k - x_j |^\beta 
 \delta \Big ( 1 - \sum_{l=1}^N x_l \Big ).
\end{multline}
The normalisation $K_N$ permits a product of gamma function evaluation generalising (\ref{7c}).

\begin{proposition}\label{P10}
Let $W_{\alpha,\beta,n}$ be given as in (\ref{7c+}), and let
  \begin{equation}\label{11a} 
  \eta_N(a,b,p,\beta)  = (b+1) p + {\beta \over 2} p (p-1) + (a+1) (N-p) + {\beta \over 2}(N-p) (N-p-1).
   \end{equation}
   We have
  \begin{equation}\label{11b}   
   K_N(a,b,p,\beta) = {W_{b,\beta,p} W_{a,\beta,N-p} \over \Gamma( \eta_N(a,b,p,\beta))}.
  \end{equation}
  \end{proposition}  
 
\begin{proof}
Due to the delta function constraint, the range of integration $[0,1]$ in each terminal of (\ref{10h})
can be replaced by $[0,\infty)$. Doing this, and introducing a parameter $t$ by the replacement
$$
  \delta \Big ( 1 - \sum_{l=1}^N x_l \Big ) \mapsto   \delta \Big ( t - \sum_{l=1}^N x_l \Big )
  $$
 to obtain the quantity $  K_N(a,b,p,\beta;t)$  allows the Laplace transform with respect to $t$ to be taken, with the result
  $$
   \hat{K}_N(a,b,p,\beta;s) = {1 \over s^{  \eta_N(a,b,p,\beta) }} W_{b,\beta,p} W_{a,\beta,N-p} .
   $$
   Now taking the inverse Laplace transform, and setting $t=1$ gives (\ref{11b}).
   \end{proof}
   
   It remains to consider the analyticity properties
   of (\ref{10fx}). In this regard the case $\beta$ even is special. For this case, in the integrand we have
     \begin{equation}\label{3.21a} 
   |1 + s (x_l - x_{l'}) |^\beta =  (1 + s (x_l - x_{l'}) )^\beta 
    \end{equation}
   revealing that $F_N^{(p)}(s)$ is an analytic function in the range $-1<s<N-p$,
   and in the cases $a \in \mathbb Z_{\ge 0}$ is in fact a polynomial.
    Noting too that for $\beta$ even (\ref{7ed}) is consistent with
   (\ref{X1}), Proposition \ref{P1x} has thus been established for $\beta$ even.
   In the next subsection it will be shown how this analytic function/polynomial
   can be computed.
   
   We now examine the analyticity properties of (\ref{10fx}) which hold true for general $\beta \ge 0$.
   One observation is that (\ref{3.21a}) is always valid for $|s| < 1$, and
 so with $b$ a non-zero integer,    $F_N^{(p)}(s)$ is analytic for $|s| < 1$. Also, with $p=0$, 
 factors  of the form (\ref{3.21a}) are absent from the integrand of (\ref{10fx}), so we have
 that $F_N^{(0)}(s)$ is analytic for $-1 < s < N$ as required by  Proposition \ref{P1x}.
 It follows that Proposition \ref{P1x} is true for all $\beta \ge 0$ in the
 restricted range $0\le t \le 2$. Taking into consideration the symmetry (\ref{7e}), and the
 facts that $F_N^{(0)}(s)$ is computable for $0 \le s \le N$ and $F_N^{(1)}(s)$ is computable for $0 \le s \le 1$,
 we therefore have a scheme to compute
 $P^{(J)}(t)$ for all $\beta \ge 0$, provided both $a,b$ are non-negative integers, and $N \le 4$
 (for $N =2$ we can do better: the formula (\ref{gf.3a}) below is valid for general $a,b > -1$ as well as
 general $\beta \ge 0$).
 
 For $\beta$ not even the factors in the integrand of the form of the LHS of (\ref{3.21a})
 are not analytic beyond $|s| < 1$. However, for $\beta$ odd there is a simple analytic continuation,
 which is given by the RHS of (\ref{3.21a}). The task that presents itself is to determine the linear combination
 of $\{ \chi_{t > p} (t - p)^{\gamma_p} F_N^{(p)}(t-p) \}_{p=0}^{N-1} \}_{p=0}^{N-1}$ which equals $P^{(J)}(t)$
 when  $F_N^{(p)}(s)$ is specified by (\ref{10fx}) modified according to (\ref{3.21a}). From the
 definition we can write
 \begin{equation*}
 P^{(J)}(t) = {N! \over S_N(a,b,\beta)} \int_{R_N} dx_1 \cdots dx_N \,
 \delta \Big ( t - \sum_{l=1}^N x_l \Big ) \prod_{l=1}^N x_l^a (1 - x_l)^b 
 \prod_{1 \le j < k \le N} (x_j - x_k),
 \end{equation*}
 where
 $$
 R_N : \: 1 > x_1 > x_2 > \cdots > x_N > 0.
 $$
 To isolate the singularity about $t=p$ ($t > p$) we  change variables
  \begin{equation}\label{Ra1}
 x_l = 1 - s y_l \: \: (l=1,\dots,p) \qquad  x_{l'} =  s y_{l'} \: \: (l'= p + 1,\dots, N),
  \end{equation} 
where $s = t - p$. For $0 < s \ll 1$ at least, this shows the functional form of the
singular term to equal
 \begin{multline}\label{Ra}
 {N! \over S_N(a,b,\beta)} \chi_{s > 0} s^{\gamma_p} (-1)^{b (p+1) + p (p-1)/2}
 \int_{\tilde{R}_N} dy_1 \cdots dy_N \, \delta \Big ( 1 - \sum_{l=1}^N y_l \Big ) \prod_{l=1}^p  (1 -  s y_l)^a y_l^b \\
 \times \prod_{l=p+1}^N y_l^a ( 1 - sy_l)^b \Delta ( \{ y_j \}_{j=1}^p   \Delta ( \{ y_j \}_{j=p+1}^N 
 \prod_{j=1}^p \prod_{k=p+1}^N (1 - s (y_l - y_{l'}).
 \end{multline}
 Here $\Delta( \{u_j \}_{j=1}^q) := \prod_{1 \le j < k \le q} (u_j - u_k)$, and
 $$
 \tilde{R}_N : \: 1 > y_1 >  \cdots > y_p > 0, \: \:    1 > y_{p+1} >  \cdots > y_N > 0.
 $$
 The reason for the sign $(-1)^{b (p+1)}$ comes from the first of the change  of variables (\ref{Ra1}) applied to
 $\prod_{l=1}^p (1 - x_l)^b$, and the change of sense of the direction of the integration domain caused by
 this change of variables. This same change  of variables, now applied to the factor in the integrand
 $\Delta( \{x_j \}_{j=1}^q)$ is responsible for the sign $(-1)^{p (p-1)/2}$. This latter sign is a key difference
 to the sign  (\ref{7ed}) which is found for $\beta$ even, as is the fact that the integral in (\ref{Ra}) is, with
 the assumption $b$ is a non-negative integer, analytic
 for $s > -1$ at least as required in the statement of Proposition \ref{P1x}.
 Symmetrising the integrand of (\ref{Ra}) in the integration variables $\{ y_l \}_{l=1}^p$ and
 $\{ y_l \}_{l=p+1}^N$ shows exact agreement with the combinatorial factor as well. 
 Hence the $\beta$ odd case of Proposition \ref{P1x} is established.

   \begin{remark}\label{R3.5}
  1.  Kumar and Pandey \cite{KP10} have used the Pfaffian and determinant structures
   present in the cases $\beta =1,4$, and 2 respectively to obtain evaluations of $\hat{P}^{(J)}(s)$
   in these cases, where attention was further restricted to requiring $b=0$, and
   $a$ to be of the form $\beta(M-N+1)/2 - 1$ with $M \ge N$.  Most significantly
   they showed that for small values of $N$ and $M-N$ the inverse Laplace transform
   could be carried out explicity (this approach was also used in earlier works
   \cite{KC72,KSS09} but without obtaining explicit results for the inverse transform
   beyond $N = 2$). One example is \cite[minor rewrite of Eq.~(23)]{KP10}
  \begin{multline}\label{11c}  
 {P}^{(J)}(t) \Big |_{a=b=0, \beta = 1,N=3} =
 {3 \over 8} \Big ( t^5 -(t-1)^3 ( 40 - 10(t-1) + (t-1)^2  ) \chi_{t>1}  \\
- (t-2)^3 ( 40 + 10 (t-2) + (t-2)^2 )  \chi_{t>2} \Big ),
\end{multline}
supported on $0 < t < 3$ (as an aside we note that this expression is unchanged
by the replacement $t \mapsto 3 - t$ as is consistent with (\ref{7e})).
Highlighting the structural features by reading off from this that
\begin{multline}\label{11d}  
 {P}^{(J)}(t) \Big |_{a=b=0, \beta = 1,N=3} =  {3 \over 8} t^5 - 15 (t-1)^3\Big ( 1 + O((t-1)) \Big )  \chi_{t>1}\\
- 15(t-2)^3 \Big  ( 1 + O((t-2)) \Big ) \chi_{t>2} \Big ),
\end{multline}
we get consistency with the same expansion as implied by Proposition \ref{P1x}. \\
2. We see from (\ref{10fx}) that for $a,b \in \mathbb Z_{\ge 0}$ and $\beta \in \mathbb Z^+$ the
quantity $F_N^{(p)}(s)$ is a polynomial in $s$ of degree
\begin{equation}\label{3.29a}
ap + b (N-p) + p (N-p) \beta,
\end{equation}
as is illustrated by (\ref{11c}).
\end{remark}

\subsection{The differential-difference system}\label{S3.3}
We now turn our attention to the power series expansion of each term $p$ in the summation of (\ref{10hy}).
For this we introduce the integrals
\begin{multline}\label{6.1+}
{H}_{p}(t) = {1 \over C_p^N}\int_0^1 dx_1 \cdots  \int_0^1 dx_N \,
\delta \Big ( t - \sum_{l=1}^N x_l \Big ) \prod_{l=1}^N x_l^{a} (1 - x_l)^{b-1}   \\
\times \prod_{1 \le j < k \le N} | x_k - x_j|^\beta 
e_p(1 - x_1,\dots, 1 - x_N),
\end{multline}
for $p=0,1,\dots,N$. Note that the Laplace transform of $H_p(t)$ gives $\hat{H}_p(x)$
as specified by (\ref{6.1}). The fact that $\{ \hat{H}_p(x) \}$ satisfies a differential-difference
system allows us to deduce a  corresponding differential-difference system
satisfied by these integrals.

\begin{corollary}\label{C6+}
Define $\tilde{B}_p, \tilde{D}_p$ as in Corollary \ref{C6}.
The multiple integrals $\{  {H}_p(t) \}_{p=0}^{N}$ satisfy the
differential-difference system
\begin{equation}\label{6.1a+}
(N - p) {d \over dt}  {H}_{p+1}(t) 
=   (N - p) {d \over dt} {H}_p(t)   + (\tilde{B}_p  - 1) {H}_p(t)  -  t  {d \over dt} {H}_p(t) -  \tilde{D}_p  {H}_{p-1}(t),
\end{equation}
valid for $p=0,\dots, N$.
\end{corollary}

\begin{proof}
Following \cite{Da70}, the strategy is to introduce the inverse Laplace transform into the 
differential-difference system (\ref{6.1a}) by multiplying through by $e^{xt}$ and integrating
over the contour Re$(x) = \gamma$. Minor manipulation leads to (\ref{6.1a+}).
\end{proof}

\begin{remark}
With  $b=0$ we have from the definition that $\tilde{D}_N = 0$, and we read off from (\ref{6.1a+})
with $p = N$ that
\begin{equation}\label{7.1a+}
 (\tilde{B}_N  - 1) {H}_N(t)  -  t  {d \over dt} {H}_N(t) = 0.
 \end{equation}
 Consequently  $H_N(t) = C t^{\tilde{B}_N  - 1}$. But $H_N(t)$ is proportional to $P_N^{(J)}(t)$, so 
 after noting the explicit value of $\tilde{B}_N$ this
 is consistent with (\ref{7d+}).
 \end{remark}
 
 \begin{remark}
 For generic $\tilde{B}_p, \tilde{D}_p$, eliminating $H_0(t), H_1(t),\dots,H_{N-1}(t)$ from the 
 system (\ref{6.1a+}) leads to a linear differential equation of degree $(N+1)$ for $H_N(t)$.
 However, since the particular values with $p=0$ are $\tilde{B}_p = \tilde{D}_p = 0$ as seen
 from Corollary \ref{C6}, the equation in (\ref{6.1a+}) with $p=0$ can be integrated to deduce
 \begin{equation}\label{gf.1a}
 N H_1(t) = (N-t) H_0(t).
  \end{equation}
With this refinement, the elimination process gives a linear differential equation of degree $N$ for
$H_N(t)$. The simplest nontrivial case is $N = 2$, where we obtain
\begin{multline}\label{gf.2a}
t(t-1)(t-2) H_2''(t) - \Big ( (4a+2\beta) - 4(2a+b+\beta)t + (3(a+b)+2\beta) t^2 \Big ) H_2'(t)  \\
+ (1 + 2a + 2b + \beta) \Big ( -2a - \beta + (a+b+\beta)t \Big ) H_2(t) = 0.
\end{multline} 
In the case $\beta = 1$ this characterisation of $H_2(t)$ has previously been given by Davis \cite[Eq.~(3.3), with
$a={1 \over 2}(n_1 - 3)$, $b={1 \over 2}(n_2 - 3)$]{Da70}. Moreover, a sequence of transformations
were identified, reducing this case of (\ref{gf.2a}) to the Gauss hypergeometric differential equation. These
can be extended to general $\beta > 0$. Thus we first substitute for $H_2(t)$ in favour of $K_2(t)$ by writing
$$
H_2(t) = t^{2a+1+\beta} (2 - t)^{2b} K_2(t)
$$
(note that the power of $t$ is consistent with (\ref{7d}) in the case $N=2$).
In the resulting equation for $K_2(t)$, we then make the change of variables $s=(t/(2-t))^2$ (this being independent of
$\beta$ is exactly the same as in \cite{Da70}). The 
Gauss hypergeometric equation for ${}_2 F_1(\lambda_1,\lambda_2;\lambda_3;s)$ results, with
$$
\lambda_1 = {1 \over 2} (\beta + 1), \quad \lambda_2 = -b, \quad \lambda_3 = a + {1 \over 2} (\beta + 3).
$$
Fixing the proportionality constant by appealing to  (\ref{7d}) we therefore conclude that for $0<t<1$
 \begin{equation}\label{gf.3a}
 P^{(J)}(t) \Big |_{N=2}  = C_{a,b,\beta} t^{2a+1+\beta} (1 - t/2)^{2b}  \, {}_2 F_1\Big ( {1 \over 2} (\beta + 1),  -b;  a + {1 \over 2} (\beta + 3); \Big ( {t \over 2 - t} \Big )^2 \Big ),  
  \end{equation}
  where 
   \begin{equation}\label{gf.3b}
C_{a,b,\beta} =     {\Gamma(a+b+2+\beta/2) \Gamma(a+b+2+\beta) \over \Gamma(2a+2+\beta) \Gamma(b+1) \Gamma(b+1+\beta/2)} .
 \end{equation}
The evaluation in the range $1<t<2$ follows from this and (\ref{7e}) with $N = 2$.
\end{remark}

\smallskip

We know that the system (\ref{6.1a}) is equivalent to the matrix differential equation (\ref{MD}).
Likewise,
introducing the $(N+1) \times (N+1)$ matrices
\begin{align}\label{AB}
\mathbf A & = - {\rm diag} \,[N,N-1,\dots,0] + {\rm diag}^+[N,N-1,\dots,1] \nonumber \\
\mathbf B & =  {\rm diag} \,[\tilde{B}_0-1, \tilde{B}_1-1,\dots,\tilde{B}_N-1] - {\rm diag}^-[\tilde{D}_1,\tilde{D}_2,\dots,\tilde{D}_N],
\end{align}
and the column vector $\mathbf H(t) = [ H_p(t) ]_{p=0}^N$ we see that (\ref{6.1a+}) is equivalent to the
first order linear matrix differential equation
\begin{equation}\label{AB+}
(\mathbf A + t \mathbf I_{N+1}) {d \over dt} \mathbf H(t) = \mathbf B \mathbf H(t).
\end{equation}
In the case $\beta = 1$ a minor rewrite of this differential equation is already present in
the pioneering paper \cite{Da70} for this class of result.
The latter also shows us how to make use of general theory in \cite{CL55} relating to
matrix differential equations, to transform (\ref{AB+}) to a form allowing the power
series expansions about the regular singular points $t=p$, ($p=0,1,\dots,N-1$)
to be analysed, and indeed exhibiting this feature of these points.

\begin{proposition}\label{P14}
With $\mathbf P =  \Big [ \binom{N-j}{N-k} \Big ]_{j,k=0}^N$, and $\mathbf H(t)$ as
specified below (\ref{AB}), define the column vector $\mathbf G(t)$ by
$\mathbf G(t) =  \mathbf P^{-1} \mathbf H(t)$. Define the tridiagonal matrix
$\mathbf X = [x_{jk} ]_{j,k=0}^N$ with entries
$$
  x_{jk} = \left \{
 \begin{array}{ll}  u_j, & k = j - 1\\
 v_j , & k = j  \\
  w_j, & k = j + 1 \\
 0, & {\rm otherwise}, \end{array}    \right.
 $$
 where
 \begin{align*}
 u_j & = -j( \beta N/2 + b) + (\beta/2)j^2 \nonumber \\
 v_j & = (a+b+1+(\beta/2)(2N-1))j - (\beta/2)(N-j) (2j+1) + (\beta N/2 + b)(N - 2j) - 1\nonumber \\
 w_j & = (a+b+1+(\beta/2)(2N-1))(j-N) + (\beta/2)(N-j)(N-j-1) + (\beta N/2 + b) (N-j),
 \end{align*}
 and set 
 $$
 \mathbf \Lambda  = - {\rm diag} [N, N-1,\dots, 1, 0 ].
 $$
 We have
 \begin{equation}\label{GG}
 (\mathbf \Lambda + t \mathbf I_{N+1}) {d \over dt} \mathbf G(t)  = \mathbf X   \mathbf G(t).
 \end{equation}

\end{proposition}

\begin{proof}
The general strategy from \cite{CL55} as applied in \cite{Da70} to the case $\beta = 1$ of
(\ref{AB+}) requires
finding the matrix of eigenvectors $\mathbf P$, and matrix of diagonal eigenvalues
$\mathbf \Lambda$, of the matrix $\mathbf A$.  We can check that
\begin{equation}\label{C5}
\mathbf \Lambda  = - {\rm diag} [N, N-1,\dots, 1, 0 ], \qquad \mathbf P = \bigg [ \binom{N-j}{N-k} \bigg ]_{j,k=0}^N.
\end{equation}
Note that $\mathbf P$ is upper triangular. We can check too that
\begin{equation}\label{c6}
 \mathbf P^{-1} = \bigg [ (-1)^{j+k} \binom{N-j}{N-k} \bigg ]_{j,k=0}^N.
 \end{equation}
 The relevance of knowing $ \mathbf P^{-1}$  is that it is required in the diagonalisation formula
 $ \mathbf P^{-1} \mathbf A \mathbf P = \mathbf \Lambda $.
 Thus if we write $\mathbf H(t) = \mathbf P \mathbf G(t)$ and multiply through by $ \mathbf P^{-1}$
 the matrix differential equation (\ref{AB+}) reads
 \begin{equation}\label{c7}
 (\mathbf \Lambda + t \mathbf I_{N+1}) {d \over dt} \mathbf G(t) =  (\mathbf P^{-1} \mathbf B \mathbf P) \mathbf G(t).
\end{equation}

From the definition of $\mathbf B$ as given by (\ref{AB}) and Corollary \ref{C6}, we see that to compute
$\mathbf P^{-1} \mathbf B \mathbf P$ it suffices to compute
$
\mathbf Q_j = \mathbf P^{-1} \mathbf R_j \mathbf P
$, $(j=1,2,3)$
where 
\begin{align*}
& \mathbf R_1 = {\rm diag} \, [0,1,\dots,N], \quad \mathbf R_2 = {\rm diag}^- \, [1,\dots,N], \\
& \mathbf R_3 =  {\rm diag} \, [0^2,1^2,\dots,N^2] - {\rm diag}^- \, [1^2,\dots,N^2].
\end{align*}
Thus
\begin{equation}\label{vB}
 \mathbf B = (a + b + 1 + (\beta/2) (2N - 1)) \mathbf R_1 + (\beta N/2 + b) \mathbf R_2 - (\beta/2) \mathbf R_3 - \mathbf I_{N+1},
\end{equation}

From the explicit form of $\mathbf P^{-1}, \mathbf P$ we see
\begin{align*}
(\mathbf Q_1)_{jk}  & =  (-1)^j \sum_{k' = 0}^N (-1)^{k'} \binom{N-j}{N-k'} k'  \binom{N-k'}{N-k} \\
(\mathbf Q_2)_{jk}  & =  (-1)^j \sum_{k' = 0}^N (-1)^{k'+1} \binom{N-j}{N-(k'+1)} (k' +1) \binom{N-k'}{N-k} \\
(\mathbf Q_3)_{jk}  & =  (-1)^j \sum_{k' = 0}^N (-1)^{k'} \bigg \{ \binom{N-j}{N-k'} (k')^2
+ \binom{N-j}{N-(k'+1)} (k' +1)^2 \bigg \}  \binom{N-k'}{N-k}.
\end{align*}
To evaluate $(\mathbf Q_1)_{jk}$, note upon simple manipulation and use
of the binomial theorem that
\begin{equation}\label{bi1}
 \sum_{k' = 0}^N (-1)^{k'} \binom{N-j}{N-k'}   \binom{N-k'}{N-k}  x^{k'} =  \chi_{k \ge j} 
 {(N-j)! \over (N-k)! (k-j)!} (-x)^j  (1 - x)^{k-j}.
 \end{equation}
 Differentiating with respect to $x$ and setting $x=1$ shows
 $$
 (\mathbf Q_1)_{jk} = \left \{
 \begin{array}{ll}  j, & k = j \\
 j - N, & k = j + 1 \\
 0, & {\rm otherwise}. \end{array}    \right.
 $$
 
 In relation to $(\mathbf Q_2)_{jk}$, we have analogous to (\ref{bi1}),
\begin{multline}\label{bi2}
 \sum_{k' = 0}^N (-1)^{k'} \binom{N-j}{N-(k'+1)}   \binom{N-k'}{N-k}  x^{k'} \\ =  \chi_{k \ge j-1} 
 {(N-j)! \over (N-k)! (k-j+1)!} (-x)^{j-1} \Big ( (N - j + 1) (1 - x)^{k-j+1} + x (k - j + 1) (1 - x)^{k-j} \Big ),
 \end{multline} 
 and hence
 $$
  (\mathbf Q_2)_{jk} = \left \{
 \begin{array}{ll}  - j, & k = j - 1\\
 N - 2j , & k = j  \\
  N - j, & k = j + 1 \\
 0, & {\rm otherwise}. \end{array}    \right.
 $$
 Knowledge of the summations (\ref{bi1}) and (\ref{bi2}) is sufficient to compute $(\mathbf Q_3)_{jk}$. We find
 $$
  (\mathbf Q_3)_{jk} = \left \{
 \begin{array}{ll} -j^2, & k = j - 1\\
 (N - j)(2j+1), & k = j  \\
 - (N-j)(N-j-1), & k = j + 1 \\
 0, & {\rm otherwise}. \end{array}    \right.
 $$
 
 Since from (\ref{vB})
 $$
 \mathbf P^{-1}  \mathbf B   \mathbf P = (a + b + 1 + (\beta/2) (2N - 1)) \mathbf Q_1 + (\beta N/2 + b) \mathbf Q_2 - (\beta/2) \mathbf Q_3 - \mathbf I_{N+1},
 $$
with the above knowledge of the entries of the matrices $\mathbf Q_i$, the differential equation (\ref{GG}) now follows from (\ref{c7}).
 \end{proof}

 For $p=0,1,\dots,N$ define
 $\mathbf \Lambda_p = {\rm diag} \, [(p)^{N+1} ] +  \mathbf \Lambda$, and set $s = t - p$, allowing us to rewrite
 (\ref{GG}) as
\begin{equation}\label{GG+} 
( \mathbf \Lambda_p + s \mathbf I_{N+1} ) {d \over ds} \mathbf G(s) = \mathbf X  \mathbf G(s).
 \end{equation}
 We have from the corresponding definitions that
 \begin{equation}\label{GH} 
 (\mathbf G(t))_N =  H_N(t)  \propto P^{(J)}(t).
 \end{equation}
 Thus in the variable $s$, we have from (\ref{10hy}) that  $G_{N}(s)$, now specified as the final
 component in the vector solution of the matrix differential equation (\ref{GG+}), permits a Frobenius type
 series expansion
  \begin{equation}\label{GH+} 
(\mathbf G(s))_N = s^{\gamma_p} \sum_{l=0}^\infty c_l s^l,
 \end{equation}
 where $\gamma_p$ is specified by (\ref{10g}). We now take up the task of showing this directly from
 (\ref{GG+}), and moreover specifying a recurrence for the coefficients $\{  c_l \}$. 
 
 \begin{proposition}\label{P16}
 Let $v_j$ be specified as in Proposition \ref{P14}. We have that (\ref{GG+}) admits a vector Frobenius type solution
  \begin{equation}\label{3.40a}
 \mathbf G(s) = C s^{v_{N-p}} \sum_{l=0}^\infty \mathbf g_l s^l.
  \end{equation}
 Here $C \ne 0$ is an arbitrary scalar, and the coefficient vectors  $\{ \mathbf g_l  \}$ are determined
 by
  \begin{equation}\label{g0}
  \mathbf g_0 = [ \delta_{j,N-p} ]_{j=0}^N
\end{equation}  
 and
  \begin{equation}\label{g1}
(v_{N-p} + l ) \mathbf \Lambda_p  \mathbf g_l = ( \mathbf X - (v_{N-p} + l - 1) \mathbf I_{N+1} )   \mathbf g_{l-1}, \qquad l=1,2,\dots
\end{equation} 
The case $l=n$ of (\ref{g1}) determines $(\mathbf g_n)_s$ for $s \ne N-p$, and the case $l=n+1$ implies
 \begin{equation}\label{g2}
 n (\mathbf g_n)_{N-p} = u_{N-p} (\mathbf g_n)_{N-p-1} +  w_{N-p} (\mathbf g_n)_{N-p+1} ,
\end{equation} 
where $u_j, w_j$ are as in  Proposition \ref{P14}.
\end{proposition}

\begin{proof}
We simply substitute the Frobenius series in (\ref{GG+}) and equate like powers of $s$ to deduce
(\ref{g1}). In the case $l=0$ the right hand side of (\ref{g1}) is to be interpreted as the zero vector. Thus
$\mathbf g_0$ is an eigenvector of $\mathbf \Lambda_p$ corresponding to the eigenvalue zero.
From the definition of $\mathbf \Lambda_p$ this eigenvector is, up to choice of normalisation, given
by (\ref{g0}).

From (\ref{g1}), taking into consideration the definition of $\mathbf \Lambda_p$ we can read off
the value of the components $(\mathbf g_l)_s$ for $s \ne N-p$. 
The case $s = N - p$ is special. Then the diagonal entry in row $N - p$ of $\lambda_p$ is
zero, implying each entry of row $N - p$  is zero so $\lambda_p$ is not invertible. On the other
hand, this feature implies that row $N - p$ on the right hand side of (\ref{g1}) must also consist of
all zeros. Choosing now $l=n+1$ in the resulting equation (\ref{g2}) results.
\end{proof}

\begin{remark}
From the definition (\ref{10g}) of $\gamma_p$, and the definition in  Proposition \ref{P14} of $v_j$
we see that
  \begin{equation}\label{g1g}
v_{N-p} - \gamma_p = - p.
\end{equation}  
This is consistent with the initial condition (\ref{g0}) that implies $(\mathbf g_l)_q = 0$ for $q=N-p+1,\dots,N$.
\end{remark}

We can use Proposition \ref{P16} to explicitly compute the Frobenius series
of $({\mathbf G}(s))_N  =  ({\mathbf G(t-p)})_N$ for $p=0,1,\dots,N$. The structure of the coefficients
becomes increasingly more complicated with increasing powers in the series, and
also as $p$ increases. As specific examples, upon normalising so that the 
coefficient of the leading term is unity we have
\begin{align*}
 ({\mathbf G(t)})_N & = t^{v_N} \Big ( 1 - {u_N w_{N-1} \over v_N + 1} t + \cdots \Big ) \\
 ({\mathbf G(t-1)})_N & = (t-1)^{v_{N-1}+1} \Big ( 1 - \Big ( {u_{N-1} w_{N-2} \over v_{N-1} + 2} -  {u_N w_{N-1} \over v_{N-2} + 2} 
-{(v_N - v_{N-1} - 1) \over v_{N-1} + 2} \Big ) ( t  - 1) + \cdots \Big ).
\end{align*}
It may be noted that for $a,b \in \mathbb Z_{\ge 0}$ and $\beta \in \mathbb Z^+$, the series in $(t-p)$ starting with 1 within the brackets,
and a factor $(t-p)^{v_{N-p}+p}$ outside, terminates  to give a polynomial of degree $ap+b(N-p)+p(N-p)\beta$ as is consistent with
Remark \ref{R3.5}, point 2.
Despite the ability to compute the Frobenius series solutions, the functional form of $P^{(J)}(t)$ remains undetermined
without knowledge of the scalar $C=C(p)$ in (\ref{3.40a}) for each $p$. Davis
\cite[Paragraph above (3.3)]{Da70} remarks: ``The calculation of the numerical
coefficients in these linear combinations presents a formidable unsolved
problem'', and leaves it there. While this remains true in general, our (\ref{10hy}) and
Proposition \ref{P10} solves this problem for parameters $b,\beta$ non-negative
integers. Thus, using the notation
 therein, under this assumption we read off the explicit  evaluation of the sought scalar
\begin{equation}\label{3.40b}
C(p) ( \mathbf g_p)_N = {1 \over S_N(a,b,\beta)} \xi_p \binom{N}{p} K_N(a,b,p,\beta),
\end{equation}
and (\ref{10hy})  follows.

For this knowledge to be of practical value, we have to overcome the problem that  the power series factor
of the Frobenius solution, denoted $F_N^{(p)}(s)$
in (\ref{10hy}), generally only has radius of convergence unity. In fact this is not an issue if we restrict the parameter $a$, in
addition to $b, \beta$, to be a non-negative integer. Then $F_N^{(p)}(s)$ is a polynomial of degree (\ref{3.29a}).

For $a \notin \mathbb Z_{\ge 0}$, a possible strategy to analytically
continue beyond $s = 1$ is to use the Frobenius series about $s=0$ to obtain an accurate value of $\mathbf G(s_0)$ with $0< s_0 < 1$
 inside the radius of convergence. The point $s=s_0$ is an ordinary point of the differential equation 
(\ref{GG+}), which with $u = s- s_0$ can be rewritten
\begin{equation}\label{GGs} 
 {d \over du} \mathbf G(u) =   ( \mathbf \Lambda_p + (u+ s_0) \mathbf I_{N+1} )^{-1} \mathbf X  \mathbf G(u).
 \end{equation}
 Given $\mathbf G(s_0)$, (\ref{GGs}) uniquely determines a power series series solution in $u$ with radius of convergence $1 + s_0$, and thus gives
 an analytic continuation relative to the Frobenius solution. Repeating this procedure, for a given value of
 the singular point $t=p$, $G(t)$ can, in theory, be extended to the required range $p \le t \le N $. Unfortunately, in practice it
 was found that this procedure was unstable with respect to the required truncations of the power series in the components of $ \mathbf G(u)$, so
 a working solution for $a \notin \mathbb Z_{\ge 0}$ remains.

\begin{figure*}
\centering
\includegraphics[width=0.98\textwidth]{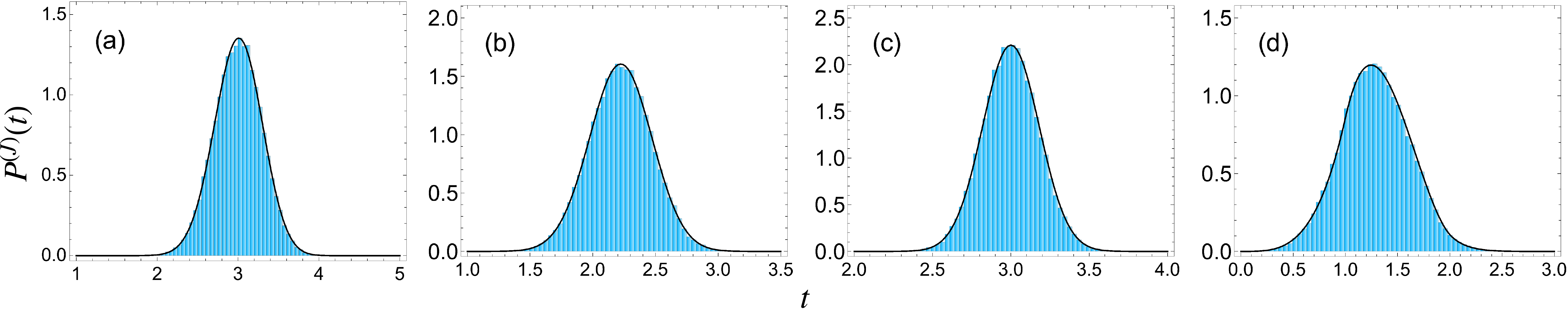}
\caption{Plots of probability density function of Jacobi-ensemble trace for three parameter sets: (a) $N = 5, a = 3, b = 1, \beta = 1$, (b) $N = 4, a = 1, b = 0, \beta = 2$, (c) $N = 6, a = 0, b = 0, \beta = 4$, (d) $N = 3, a = -1/2, b = 0, \beta = 1$.}
\label{fig1}
\end{figure*}

For $a \in  \mathbb Z_{\ge 0}$, the recursive scheme of Proposition \ref{P10} combined with (\ref{10hy})  for
 evaluating the trace-distribution in the Jacobi case has been implemented in Mathematica codes and the corresponding files have been provided as ancillary material. Additionally, for interested readers, we include an ``experimental'' code which implements analytical continuation about the ordinary point $s_0=1/2$ to deal with $a \notin \mathbb Z_{\ge 0}$ case. As cautioned above, this procedure is unstable and may produce satisfactory results only for certain choices of parameters. A few example plots of the evaluated distribution and their comparison with Monte Carlo simulations based on $\beta$-Jacobi matrix model~\cite{ES2008} appear in figure~\ref{fig1}. In particular, plot (d) shows the result for $a=-1/2$, along with $n=3, b=0, \beta=1$). It has been obtained using the analytical continuation approach for $a \notin  \mathbb Z_{\ge 0}$, which works satisfactorily for this set of parameters. Furthermore, on this resolution, the large $N$
Gaussian form of $P^{(J)}(t)$, with specific means and variances (see e.g.~\cite[Eqns.~(42)--(44)]{KP10}), is already evident despite the values of $N$ being small and the exact functional
form of  $P^{(J)}(t)$ changing at $t=1,2,\dots,N-1$.

\section*{Acknowledgements}
 P.J.F.~acknowledges support from the Australian Research Council
 (ARC) through the ARC Centre of Excellence for Mathematical \& Statistical Frontiers.

 \providecommand{\bysame}{\leavevmode\hbox to3em{\hrulefill}\thinspace}
\providecommand{\MR}{\relax\ifhmode\unskip\space\fi MR }
\providecommand{\MRhref}[2]{%
  \href{http://www.ams.org/mathscinet-getitem?mr=#1}{#2}
}
\providecommand{\href}[2]{#2}

\end{document}